\documentclass[a4paper,11pt]{article}

\usepackage{amssymb}
\usepackage{wrapfig}
\usepackage{url}
\usepackage{alltt}
\usepackage{hhline}
\usepackage{multirow}
\usepackage{graphicx}
\usepackage{amsmath}
\usepackage{wrapfig}
\usepackage[algoruled,linesnumbered,algo2e,noend]{algorithm2e}
\usepackage{theorem}
\usepackage{authblk}
\usepackage{newproof}

\newcommand{\rev}[1]{{#1}^R}

\newcommand{\cnt}{\mathrm{cnt}}

\newtheorem{theorem}{Theorem}
\newtheorem{lemma}{Lemma}

{\theorembodyfont{\normalfont}
  \newtheorem{problem}{Problem}
  \newtheorem{example}{Example}
}

\title{A hardness result and new algorithm for the longest common palindromic subsequence problem}

\date{}

\author[$\dagger$]{Shunsuke Inenaga}
\author[$\ddagger$]{Heikki Hyyr\"o}

\affil[$\dagger$]{Department of Informatics, Kyushu University, Japan\\
\texttt{inenaga@inf.kyushu-u.ac.jp}}

\affil[$\ddagger$]{School of Information Sciences, University of Tampere, Finland\\
\texttt{heikki.hyyro@uta.fi}}

\begin{document}

\maketitle

\begin{abstract}
  The \emph{$2$-LCPS problem}, first introduced by Chowdhury et al.
  [Fundam. Inform., 129(4):329--340, 2014], asks one to compute
  (the length of) a longest palindromic common subsequence
  between two given strings $A$ and $B$.
  We show that the $2$-LCPS problem is at least as hard as
  the well-studied longest common subsequence problem for 4 strings.
  Then, we present a new algorithm which solves the $2$-LCPS problem
  in $O\big(\sigma M^2 + n \big)$ time,
  where $n$ denotes the length of $A$ and $B$,
  $M$ denotes the number of matching positions between $A$ and $B$,
  and $\sigma$ denotes the number of distinct characters occurring
  in both $A$ and $B$.
  Our new algorithm is faster than Chowdhury et al.'s sparse algorithm
  when $\sigma = o(\log^2n \log\log n)$.
\end{abstract}

%%%%%%%%%%%%%%%%%%%%%%%%%%%%%%%%%%%%%%%%%%%%%%%%%%%%%%%%%%%%%%%%%%%

\section{Introduction}

Given $k \geq 2$ string, the \emph{longest common subsequence problem} for $k$ strings
(\emph{$k$-LCS problem} for short) asks to compute (the length of) a longest string that
appears as a subsequence in all the $k$ strings. Whilst the problem is known to be NP-hard
for arbitrary many strings~\cite{Maier78}, it can be solved in polynomial time for a constant
number of strings (namely, when $k$ is constant).

The $2$-LCS problem that concerns two strings is the most basic, but also the most widely studied and used,
form of longest common subsequence computation.
Indeed, the $2$-LCS problem and similar two-string variants are central topics
in theoretical computer science and have applications e.g. in computational biology,
spelling correction, optical character recognition and file versioning.
The fundamental solution to the $2$-LCS problem is based on dynamic programming~\cite{WagnerF74} and
takes $O(n^2)$ for two given strings of length $n$\footnote{For simplicity, we assume that input strings are
of equal length $n$. However, all algorithms mentioned and proposed in this paper are applicable for strings of different lengths.}.
Using the so-called ``Four Russians'' technique~\cite{Artazarov70},
one can solve the $2$-LCS problem for strings over a constant alphabet
in $O(n^2 / \log^2 n)$ time~\cite{MasekP80}.
For a non-constant alphabet,
the $2$-LCS problem can be solved in $O(n^2 \log\log n / \log^2 n)$ time~\cite{Grabowski16}.
Despite much effort, these have remained as the best known algorithms
to the $2$-LCS problem, and no strongly sub-quadratic time $2$-LCS algorithm
is known.
Moreover, the following conditional lower bound for the $2$-LCS problem
has been shown:
For any constant $\lambda > 0$,
an $O(n^{2-\lambda})$-time algorithm which solves the $2$-LCS problem
over an alphabet of size $7$
refutes the so-called strong exponential time hypothesis (SETH)~\cite{AbboudBW15}.

In many applications it is reasonable to incorporate additional constraints
to the LCS problem (see e.g.~\cite{ChinSFHK04,Arslan07,IliopoulosR08a,KucherovPZ11,Deorowicz12,FarhanaR12,ZhuW13,FarhanaR15,ZhuWW16,ZhuWW16a}).
Along this line of research, Chowdhury et al.~\cite{ChowdhuryHIR14} introduced
the \emph{longest common palindromic subsequence problem}
for two strings (\emph{$2$-LCPS problem} for short),
which asks one to compute (the length of) a longest common subsequence
between strings $A$ and $B$ with the additional constraint that the subsequence must be a palindrome.
The problem is equivalent to finding (the length of) a longest palindrome
that appears as a subsequence in both strings $A$ and $B$,
and is motivated for biological sequence comparison~\cite{ChowdhuryHIR14}.
Chowdhury et al. presented two algorithms for solving the $2$-LCPS problem. The first is a conventional dynamic programming algorithm that runs in $O(n^4)$ time and space. The
second uses sparse dynamic programming and runs in $O(M^2 \log^2 n \log \log n + n)$ time and $O(M^2)$ space\footnote{The original time bound claimed in~\cite{ChowdhuryHIR14} is $O(M^2 \log^2 n \log \log n)$, since they assume that the matching position pairs are already computed. For given strings $A$ and $B$ of length $n$ each over an integer alphabet of polynomial size in $n$, we can compute all matching position pairs of $A$ and $B$ in $O(M+n)$ time.},
where $M$ is the number of matching position pairs between $A$ and $B$.

The contribution of this paper is two-folds:
Firstly, we show a tight connection between
the $2$-LCPS problem and the $4$-LCS problem
by giving a simple linear-time reduction from the $4$-LCS problem
to the $2$-LCPS problem.
This means that the $2$-LCPS problem is at least as hard as the $4$-LCS problem,
and thus achieving a significant improvement on
the $2$-LCPS problem implies a breakthrough on the well-studied
$4$-LCS problem, to which all existing solutions~\cite{Itoga81,HsuD84,IrvingF92,HakataI92,WangKS11} require at least $O(n^4)$ time in the worst case.
Secondly, we propose a new algorithm for the $2$-LCPS problem which 
runs in $O(\sigma M^2 + n)$ time and uses $O(M^2 + n)$ space,
where $\sigma$ denotes the number of distinct characters occurring in
both $A$ and $B$.
We remark that our new algorithm is faster than Chowdhury et al.'s
sparse algorithm with $O(M^2 \log^2 n \log \log n + n)$ running time~\cite{ChowdhuryHIR14} when $\sigma = o(\log^2 n \log\log n)$.

%%%%%%%%%%%%%%%%%%%%%%%%%%%%%%%%%%%%%%%%%%%%%%%%%%%%%%%%%%%%%%%%%%%

\section{Preliminaries} \label{sec:prelim}

\subsection{Strings}

Let $\Sigma$ be an \emph{alphabet}.
An element of $\Sigma$ is called a \emph{character}
and that of $\Sigma^*$ is called a \emph{string}.
For any string $A = a_1 a_2 \cdots a_n$ of length $n$, 
$|A|$ denotes its length, that is, $|A| = n$. 

For any string $A = a_1 \cdots a_m$,
let $\rev{A}$ denote the reverse string of $A$,
namely, $\rev{A} = a_m \cdots a_1$.
A string $P$ is said to be a \emph{palindrome}
iff $P$ reads the same forward and backward, namely,
$P = \rev{P}$.

A string $S$ is said to be a \emph{subsequence} of another string $A$
iff there exist increasing positions $1 \leq i_1 < \cdots < i_{|S|} \leq |A|$
in $A$ such that $S = a_{i_1} \cdots a_{i_{|S|}}$.
In other words, $S$ is a subsequence of $A$ iff
$S$ can be obtained by removing zero or more characters from $A$.

A string $S$ is said to be a \emph{common subsequence}
of $k$ strings ($k \geq 2$) iff
$S$ is a subsequence of all the $k$ strings.
$S$ is said to be a \emph{longest common subsequence} (\emph{LCS})
of the $k$ strings iff other common subsequences of the $k$ strings
are not longer than $S$.
The problem of computing (the length of) an LCS of $k$ strings
is called the \emph{$k$-LCS problem}.

A string $P$ is said to be a \emph{common palindromic subsequence}
of $k$ strings ($k \geq 2$) iff
$P$ is a palindrome and is a subsequence of all these $k$ strings.
$P$ is said to be a \emph{longest common palindromic subsequence} (\emph{LCPS})
of the $k$ strings iff other common palindromic subsequences of
the $k$ strings are not longer than $P$.

In this paper, we consider the following problem:
\begin{problem}[The $2$-LCPS problem]
  Given two strings $A$ and $B$,
  compute (the length of) an LCPS of $A$ and $B$.
\end{problem}

For two strings $A = a_1 \cdots a_n$ and $B = b_1 \cdots b_n$,
an ordered pair $(i, j)$ with $1 \leq i, j \leq n$
is said to be a \emph{matching position pair}
between $A$ and $B$ iff $a_i = b_j$.
Let $M$ be the number of matching position pairs between $A$ and $B$.
We can compute all the matching position pairs in $O(n + M)$ time
for strings $A$ and $B$ over integer alphabets of polynomial size in $n$.

%%%%%%%%%%%%%%%%%%%%%%%%%%%%%%%%%%%%%%%%%%%%%%%%%%%%%%%%%%%%%%%%%%%

\section{Reduction from $4$-LCS to $2$-LCPS}

In this section, we show that the $2$-LCPS problem
is at least as hard as the $4$-LCS problem.

\begin{theorem}
  The $4$-LCS problem can be reduced to the $2$-LCPS problem
  in linear time.
\end{theorem}

\begin{proof}
Let $A$, $B$, $C$, and $D$ be 4 input strings for the $4$-LCS problem.
We wish to compute an LCS of all these 4 strings.
For simplicity, assume $|A| = |B| = |C| = |D| = n$.
We construct two strings $X = \rev{A}ZB$ and $Y = \rev{C}ZD$
of length $4n+1$ each,
where $Z = \$^{2n+1}$
and $\$$ is a single character which does not appear in
$A$, $B$, $C$, or $D$.
Then, since $Z$ is a common palindromic subsequence of $X$ and $Y$,
and since $|Z| = 2n+1$ while $|A|+|B| = |C|+|D| = 2n$,
any LCPS of $X$ and $Y$ must be at least $2n+1$ long
containing $Z$ as a substring.
This implies that the alignment for any LCPS of $X$ and $Y$
is enforced so that the two $Z$'s in $X$ and $Y$ are fully aligned.
Since any LCPS of $X$ and $Y$ is a palindrome,
it must be of form $\rev{T}ZT$,
where $T$ is an LCS of $A$, $B$, $C$, and $D$.
Thus, we can solve the $4$-LCS problem by solving
the $2$-LCPS problem.
\end{proof}

\begin{example}  
  Consider 4 strings
  $A = \mathtt{aabbccc}$,
  $B = \mathtt{aabbcaa}$,
  $C = \mathtt{aaabccc}$,
  and $D = \mathtt{abcbbbb}$ of length 7 each.
  Then, an LCPS of $X = \mathtt{cccbbaa}\$^{15}\mathtt{aabbcaa}$
  and $Y = \mathtt{cccbaaa}\$^{15}\mathtt{abcbbbb}$
  is $\mathtt{cba}\$^{15}\mathtt{abc}$,
  which is obtained by e.g., the following alignment:
  \vspace*{-10mm}
  \begin{figure}[h!]
    \centerline{
      \includegraphics[scale=0.5]{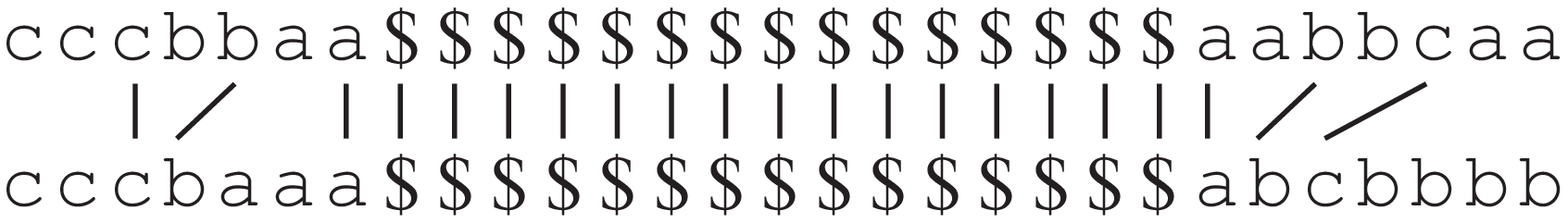}
    }
  \vspace*{3mm}
  \noindent Observe that $\mathtt{abc}$ is an LCS of
  $A$, $B$, $C$, and $D$.
  \end{figure}
\end{example}

%%%%%%%%%%%%%%%%%%%%%%%%%%%%%%%%%%%%%%%%%%%%%%%%%%%%%%%%%%%%%%%%%%%

\section{A new algorithm for $2$-LCPS}

In this section, we present a new algorithm for the $2$-LCPS problem.

\subsection{Finding rectangles with maximum nesting depth}

Our algorithm follows the approach used in the sparse dynamic programming algorithm by Chowdhury et al.~\cite{ChowdhuryHIR14}:
They showed that the $2$-LCPS problem can be reduced to a geometry problem
called the \emph{maximum depth nesting rectangle structures} problem
(\emph{MDNRS} problem for short),
defined as follows:
\begin{problem}[The MDNRS problem] \label{prob:MDNRS} \leavevmode \par
  \vspace*{2mm}
  \noindent \textbf{Input:} A set of integer points $(i, k)$ on a 
2D grid,  where each point is associated with a color $c \in \Sigma$.
  The color of a point $(i, k)$ is denoted by $c_{i, k}$.

  \vspace*{2mm}
  \noindent \textbf{Output:} A largest sorted list $L$ of pairs of points, such that
    \begin{enumerate}
     \item For any $\langle (i, k), (j, \ell) \rangle \in L$, $c_{i, j} = c_{j, \ell}$, and
     \item For any two adjacent elements $\langle (i, k), (j, \ell) \rangle$ and $\langle (i', k'), (j', \ell')$ in $L$, $i' > i$, $k' > k$, $j' < j$, and $\ell' < \ell$.
    \end{enumerate}
\end{problem}
Consider two points $(i, k)$, $(j, \ell)$ in the grid
such that $i < j$ and $k < \ell$ (see also Figure~\ref{fig:nested_rectangles}).
Imagine a rectangle defined by taking 
$(i, k)$ as its lower-left corner
and $(j, \ell)$ as its upper-right corner.
Clearly, this rectangle can be identified as the pair 
$\langle (i, k), (j, \ell) \rangle$ of points.
Now, suppose that $i$ and $k$ are positions of one input string $A = a_1 \cdots a_m$
and $j$ and $\ell$ are positions of the other input string $B = b_1 \cdots b_n$
for the $2$-LCPS problem.
Then, the first condition $c_{i, j} = c_{j, \ell}$
for any element in $L$ implies that $a_i = a_j = b_k = b_\ell$, namely,
$i, j, k, \ell$ are matching positions in $A$ and $B$.
Meanwhile, the second condition $i' > i$, $k' > k$, $j' < j$, and $\ell' < \ell$
implies that $i', j', k', \ell'$ are matching positions that are
``inside'' $i, j, k, \ell$.
Hence if we define the set of 2D points $(i, k)$ to consist of the set of matching
position pairs between $A$ and $B$ and then solve the MDNRS problem, the solution
list $L$ describes a set of rectangles with maximum nesting depth, and the characters
that correspond to the lower-left and upper-right corner matching position pairs define
an LCPS between the input strings $A$ and $B$. Recall that $M$ is the number of such pairs.
As here the lower-left and upper-right corners of each rectangle corresponding to matching position pairs,
the overall number of unique rectangles in this type of MDNRS problem is $O(M^2)$.

\begin{figure}[htb]
  \centerline{
    \includegraphics[scale=0.8]{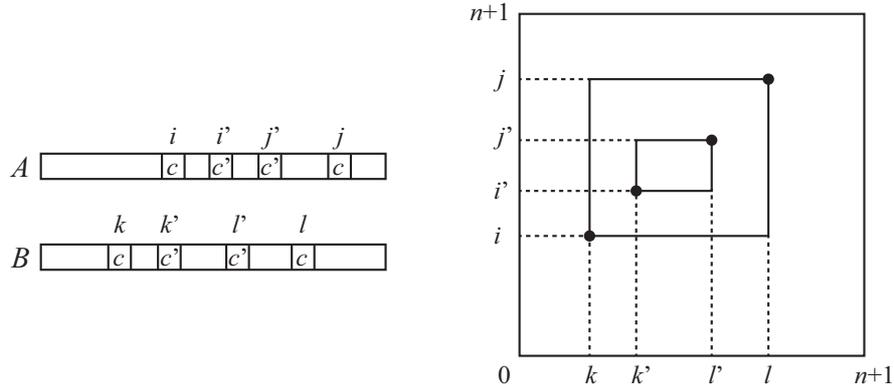}
  }
  \caption{Illustration for the relationship between
    the $2$-LCPS problem and the MDNRS problem.
    The two nesting rectangles defined by
    $\langle (i,k), (j, \ell) \rangle$ and
    $\langle (i',k'), (j', \ell') \rangle$ correspond to 
    a common palindromic subsequence $cc'c'c$ of $A$ and $B$,    
    where $c = c_{i,k} = c_{j,\ell}$ and $c' = c_{i',k'} = c_{j',\ell'}$.
  }
  \label{fig:nested_rectangles}
\end{figure}  

\subsection{Our new algorithm}

Consider the MDNRS over the set of 2D points $(i, k)$ defined by the matching
position pairs between $A$ and $B$, as described above.

The basic strategy of our algorithm is to process
from larger rectangles to smaller ones.
Given a rectangle $R = \langle (i, k), (j, \ell) \rangle$, we locate
for each character $c \in \Sigma$ a maximal sub-rectangle
$\langle (i', k'), (j', \ell') \rangle$ in $R$ that is associated to character $c$
(namely, $c_{i',k'} = c_{j',\ell'} = c$).
The following lemma is important:
\begin{lemma} \label{lem:key_lemma}
  For any character $c \in \Sigma$,
  its maximal sub-rectangle is unique (if it exists).
\end{lemma}

\begin{proof}
  Assume on the contrary that there are two distinct
  maximal sub-rectangles $\langle (i', k'), (j', \ell') \rangle$
  and $\langle (i'', k''), (j'', \ell'') \rangle$
  both of which are associated to character $c$.
  Assume w.o.l.g. that $i' > i''$, $k' < k''$, $j' < j''$ and $\ell'' > \ell'$.
  Then, there is a larger sub-rectangle 
  $\langle (i'', k'), (j', \ell'') \rangle$ of $R$
  which contains both of the above rectangles, a contradiction.
  Hence, for any character $c$, a maximal sub-rectangle in $R$
  is unique if it exists.
\end{proof}

Lemma~\ref{lem:key_lemma} permits us to define the following recursive algorithm
for the MDNRS problem:

We begin with the initial virtual
rectangle $\langle (0, 0), (n+1, n+1) \rangle$.
Suppose we are processing a rectangle $R$.
For each character $c \in \Sigma$,
we compute its maximal sub-rectangle $R_c$ in $R$
and recurse into $R_c$ until we meet one of the following conditions:
\begin{enumerate}
  \item[(1)] There remains only a single point in $R_c$,
  \item[(2)] There remains no point in $R_c$, or
  \item[(3)] $R_c$ is already processed.
\end{enumerate}
The recursion depth clearly corresponds to the rectangle nesting depth, and
we associate each $R$ with its maximum nesting depth $d_R$.
Whenever we meet a rectangle $R_c$ with Condition (3),
we do not recurse inside $R_c$ but simply return the
already-computed maximum nesting depth $d_{R_c}$.

Initially, every rectangle $R$ is marked non-processed,
and it gets marked processed as soon as the recursion for $R$ is finished 
and $R$ receives its maximum nesting depth. Each already processed rectangle
remains marked processed until the end of the algorithm.

\begin{theorem} \label{theo:linear-space-algo}
  Given two strings $A$ and $B$ of length $n$
  over an integer alphabet of polynomial size in $n$,
  we can solve the MDNRS problem (and hence the $2$-LCPS problem)
  in $O(\sigma M^2 + n)$ time and $O(M^2 + n)$ space, where $\sigma$ denotes
  the number of distinct characters occurring in both $A$ and $B$.
\end{theorem}

\begin{proof}
  To efficiently perform the above recursive algorithm, 
  we conduct the following preprocessing
  (alphabet reduction) and construct the two following data structures.
  
  \noindent \textbf{Alphabet reduction:}
  First, we reduce the alphabet size as follows.
  We radix sort the original characters in $A$ and $B$,
  and replace each original character by its rank in the sorted order.
  Since the original integer alphabet is of polynomial size in $n$,
  the radix sort can be implemented with $O(1)$ number of bucket sorts,
  taking $O(n)$ total time.
  This way, we can treat $A$ and $B$ as strings over an alphabet $[1, 2n]$.
  Further, we remove all characters that occur only in $A$ from $A$,
  and remove all characters that occur only in $B$ from $B$.
  Let $\hat{A} = \hat{a}_1 \cdots \hat{a}_{\hat{m}}$ and
  $\hat{B} = \hat{b}_1 \cdots \hat{b}_{\hat{n}}$ be
  the resulting strings, respectively.
  It is clear that we can compute $\hat{A}$ and $\hat{B}$ in $O(n)$ time.
  The key property of the shrunk strings $\hat{A}$ and $\hat{B}$
  is that since all $M$ matching position pairs in the original strings
  $A$ and $B$ are essentially preserved in $\hat{A}$ and $\hat{B}$,
  it is enough to work on strings $\hat{A}$ and $\hat{B}$
  to solve the original problem.
  If $\sigma$ is the number of distinct characters
  occurring in \emph{both} $A$ and $B$,
  then $\hat{A}$ and $\hat{B}$ are strings over alphabet $[1, \sigma]$.
  It is clear that $\sigma \leq \min\{\hat{m}, \hat{n}\} \leq n$.

  \noindent \textbf{Data structure for finding next maximal sub-rectangles:}
  For each character $c \in [1, \sigma]$,
  let $\mathcal{P}_{\hat{A},c}$ and $\mathcal{P}_{\hat{B},c}$ be the set of positions of
  $\hat{A}$ and $\hat{B}$ which match $c$,
  namely, $\mathcal{P}_{\hat{A}, c} = \{i \mid a_i = c, 1 \leq i \leq \hat{m}\}$
  and $\mathcal{P}_{\hat{B}, c} = \{k \mid b_k = c, 1 \leq k \leq \hat{n}\}$.
  Then, given a rectangle $R$,
  finding the maximal sub-rectangle $R_c$ for character $c$
  reduces to two predecessor and two successor
  queries on $\mathcal{P}_{\hat{A},c}$ and $\mathcal{P}_{\hat{B},c}$.
  We use two tables of size $\sigma \times \hat{m}$ each,
  which answer predecessor/successor queries on $\hat{A}$ in $O(1)$ time.
  Similarly, we use two tables of size $\sigma \times \hat{n}$ each,
  which answer predecessor/successor queries on $\hat{B}$ in $O(1)$ time.
  Such tables can easily be constructed in $O(\sigma (\hat{m}+\hat{n}))$ time
  and occupy $O(\sigma (\hat{m}+\hat{n}))$ space.
  Notice that for any position $i$ in $\hat{A}$
  there exists a matching position pair $(i, k)$
  for some position $k$ in $\hat{B}$, and vice versa.
  Therefore, we have $\max\{\hat{m}, \hat{n}\} \leq M$.  
  Since $\sigma \leq \min\{\hat{m}, \hat{n}\} \leq \max\{\hat{m}, \hat{n}\}$,
  we have $\sigma (\hat{m} +\hat{n}) = O(M^2)$.
  Hence the data structure occupies $O(M^2)$ space
  and can be constructed in $O(M^2)$ time.

  \noindent \textbf{Data structure for checking already processed rectangles:}
  To construct a space-efficient data structure
  for checking if a given rectangle is already processed or not,
  we here associate each character $\hat{A}$ and $\hat{B}$
  with the following character counts:
  For any position $i$ in $\hat{A}$,
  let $\cnt_{\hat{A}}(i) = |\{i' \mid \hat{a}_{i'} = \hat{a}_{i}, 1 \leq i' \leq i\}|$ and for any position $k$ in $\hat{B}$,
  let 
  $\cnt_{\hat{B}}(k) = |\{k' \mid \hat{B}_{k'} = \hat{B}_{k}, 1 \leq k' \leq k\}|$.
  For each character $c \in [1, \sigma]$,
  let $M_c$ denotes the number of matching position pairs between
  $\hat{A}$ and $\hat{B}$ for character $c$.
  We maintain the following table $T_c$ of size $M_c \times M_c$:
  For any two matching positions pairs $(i, k)$ and $(j, \ell)$
  for character $c$ (namely, $\hat{a}_{i} = \hat{b}_{k} = \hat{a}_{j} = \hat{b}_{\ell} = c$),
  we set $T_c[\cnt_{\hat{A}}(i), \cnt_{\hat{B}}(k), \cnt_{\hat{A}}(j), \cnt_{\hat{A}}(\ell)] = 0$ if the corresponding rectangle $\langle (i, k), (j, \ell) \rangle$ is non-processed,
  and set $T_c[\cnt_{\hat{A}}(i), \cnt_{\hat{B}}(k), \cnt_{\hat{A}}(j), \cnt_{\hat{A}}(\ell)] = 1$ if the corresponding rectangle is processed.
  Clearly, this table tells us whether a given rectangle is processed or not
  in $O(1)$ time.
  The total size for these tables is $\sum_{c \in [1, \sigma]}M_c^2 = O(M^2)$. 

  We are now ready to show the complexity of our recursive algorithm.
 
  \noindent \textbf{Main routine:}
  A unique visit to a non-processed rectangle can be charged to itself.
  On the other hand, each distinct visit to a processed rectangle $R$
  can be charged to the corresponding rectangle which contains $R$
  as one of its maximal sub-rectangles.
  Since we have $O(M^2)$ rectangles,
  the total number of visits of the first type is $O(M^2)$.
  Also, since we visit at most $\sigma$ maximal sub-rectangles for 
  each of the $M^2$ rectangles,
  the total number of visits of the second type is $O(\sigma M^2)$.
  Using the two data structures described above,
  we can find each maximal sub-rectangle in $O(1)$ time
  and can check if it is already processed or not in $O(1)$ time.
  For each rectangle after recursion, 
  it takes $O(\sigma)$ time to calculate the maximum nesting depth from
  all of its maximal sub-rectangles.
  Thus, the main routine of our algorithm takes a total of $O(\sigma M^2)$ time.
  
  Overall, our algorithm takes $O(\sigma M^2+ n)$ time and
  uses $O(M^2 + n)$ space.
\end{proof}

%%%%%%%%%%%%%%%%%%%%%%%%%%%%%%%%%%%%%%%%%%%%%%%%%%%%%%%%%%%%%%%%%%%

\section{Conclusions and further work}

In this paper, we studied the problem of finding
a longest common palindromic subsequence of two given strings,
which is called the $2$-LCPS problem.
We proposed a new algorithm which solves
the $2$-LCPS problem in $O(\sigma M^2 + n)$ time and $O(M^2 + n)$ space,
where $n$ denotes the length of two given strings $A$ and $B$,
$M$ denotes the number of matching position pairs of $A$ and $B$,
and $\sigma$ denotes the number of distinct characters occurring
in both $A$ and $B$.

Since the $2$-LCPS problem is at least as hard as
the well-studied $4$-LCS problem,
and since any known solution to the $4$-LCS problem takes
at least $O(n^4)$ time in the worst case,
it seems a big challenge to solve the $2$-LCPS problem
in $O(M^{2-\lambda})$ or $O(n^{4-\lambda})$ time for any constant $\lambda > 0$.
This view is supported by the recent result on a conditional lowerbound
for the $k$-LCS problem:
If there exists a constant $\lambda > 0$
and an integer $k \geq 2$ such that the $k$-LCS problem over
an alphabet of size $O(k)$ can be solved in $O(n^{k-\lambda})$ time,
then the famous SETH (strong exponential time hypothesis)
fails~\cite{AbboudBW15}.

We also remark that our method should have a good expected performance.
Consider two random strings $A$ and $B$ of length $n$ each
over an alphabet of size $\sigma$.
Since roughly every $\sigma$-th character matches between $A$ and $B$,
we have $M = O(n^2 / \sigma)$.
Hence our method runs in $O(\sigma M^2 + n) = O(n^4 / \sigma)$ expected time.
On the other hand, 
the conventional dynamic programming algorithm
of Chowdhury et al.~\cite{ChowdhuryHIR14} takes $\Theta(n^4)$ time
for \emph{any} input strings of length $n$ each.
Thus, our method achieves a $\sigma$-factor speed-up in expectation.

As an open problem, we are interested in whether
the space requirement of our algorithms can be reduced,
as this could be of practical importance.

%%%%%%%%%%%%%%%%%%%%%%%%%%%%%%%%%%%%%%%%%%%%%%%%%%%%%%%%%%%%%%%%%%%

\bibliographystyle{abbrv}
\bibliography{ref}

\begin{thebibliography}{10}

\bibitem{AbboudBW15}
A.~Abboud, A.~Backurs, and V.~V. Williams.
\newblock Tight hardness results for {LCS} and other sequence similarity
  measures.
\newblock In {\em FOCS 2015}, pages 59--78, 2015.

\bibitem{Artazarov70}
V.~Arlazarov, E.~Dinic, M.~Kronrod, and I.~Faradzev.
\newblock On economical construction of the transitive closure of a directed
  graph.
\newblock {\em Soviet Math. Dokl.}, 11:1209–--1210, 1970.

\bibitem{Arslan07}
A.~N. Arslan.
\newblock Regular expression constrained sequence alignment.
\newblock {\em J. Discrete Algorithms}, 5(4):647--661, 2007.

\bibitem{ChinSFHK04}
F.~Y.~L. Chin, A.~D. Santis, A.~L. Ferrara, N.~L. Ho, and S.~K. Kim.
\newblock A simple algorithm for the constrained sequence problems.
\newblock {\em Inf. Process. Lett.}, 90(4):175--179, 2004.

\bibitem{ChowdhuryHIR14}
S.~R. Chowdhury, M.~M. Hasan, S.~Iqbal, and M.~S. Rahman.
\newblock Computing a longest common palindromic subsequence.
\newblock {\em Fundam. Inform.}, 129(4):329--340, 2014.

\bibitem{Deorowicz12}
S.~Deorowicz.
\newblock Quadratic-time algorithm for a string constrained {LCS} problem.
\newblock {\em Inf. Process. Lett.}, 112(11):423--426, 2012.

\bibitem{FarhanaR12}
E.~Farhana and M.~S. Rahman.
\newblock Doubly-constrained {LCS} and hybrid-constrained {LCS} problems
  revisited.
\newblock {\em Inf. Process. Lett.}, 112(13):562--565, 2012.

\bibitem{FarhanaR15}
E.~Farhana and M.~S. Rahman.
\newblock Constrained sequence analysis algorithms in computational biology.
\newblock {\em Inf. Sci.}, 295:247--257, 2015.

\bibitem{Grabowski16}
S.~Grabowski.
\newblock New tabulation and sparse dynamic programming based techniques for
  sequence similarity problems.
\newblock {\em Discrete Applied Mathematics}, 212:96--103, 2016.

\bibitem{HakataI92}
K.~Hakata and H.~Imai.
\newblock The longest common subsequence problem for small alphabet size
  between many strings.
\newblock In {\em ISAAC 1992}, pages 469--478, 1992.

\bibitem{HsuD84}
W.~J. Hsu and M.~W. Du.
\newblock Computing a longest common subsequence for a set of strings.
\newblock {\em {BIT}}, 24(1):45--59, 1984.

\bibitem{IliopoulosR08a}
C.~S. Iliopoulos and M.~S. Rahman.
\newblock New efficient algorithms for the {LCS} and constrained {LCS}
  problems.
\newblock {\em Inf. Process. Lett.}, 106(1):13--18, 2008.

\bibitem{IrvingF92}
R.~W. Irving and C.~Fraser.
\newblock Two algorithms for the longest common subsequence of three (or more)
  strings.
\newblock In {\em CPM 1992}, pages 214--229, 1992.

\bibitem{Itoga81}
S.~Y. Itoga.
\newblock The string merging problem.
\newblock {\em {BIT}}, 21(1):20--30, 1981.

\bibitem{KucherovPZ11}
G.~Kucherov, T.~Pinhas, and M.~Ziv{-}Ukelson.
\newblock Regular language constrained sequence alignment revisited.
\newblock {\em Journal of Computational Biology}, 18(5):771--781, 2011.

\bibitem{Maier78}
D.~Maier.
\newblock The complexity of some problems on subsequences and supersequences.
\newblock {\em J. ACM}, 25(2):322--336, 1978.

\bibitem{MasekP80}
W.~J. Masek and M.~Paterson.
\newblock A faster algorithm computing string edit distances.
\newblock {\em J. Comput. Syst. Sci.}, 20(1):18--31, 1980.

\bibitem{WagnerF74}
R.~A. Wagner and M.~J. Fischer.
\newblock The string-to-string correction problem.
\newblock {\em J. ACM}, 21(1):168--173, 1974.

\bibitem{WangKS11}
Q.~Wang, D.~Korkin, and Y.~Shang.
\newblock A fast multiple longest common subsequence {(MLCS)} algorithm.
\newblock {\em {IEEE} Trans. Knowl. Data Eng.}, 23(3):321--334, 2011.

\bibitem{ZhuW13}
D.~Zhu and X.~Wang.
\newblock A simple algorithm for solving for the generalized longest common
  subsequence {(LCS)} problem with a substring exclusion constraint.
\newblock {\em Algorithms}, 6(3):485--493, 2013.

\bibitem{ZhuWW16}
D.~Zhu, Y.~Wu, and X.~Wang.
\newblock An efficient algorithm for a new constrained {LCS} problem.
\newblock In {\em ACIIDS 2016}, pages 261--267, 2016.

\bibitem{ZhuWW16a}
D.~Zhu, Y.~Wu, and X.~Wang.
\newblock An efficient dynamic programming algorithm for {STR-IC-STR-EC-LCS}
  problem.
\newblock In {\em GPC 2016}, pages 3--17, 2016.

\end{thebibliography}

\end{document}